\documentclass[11pt,a4paper,reqno]{amsart} 
\usepackage{amssymb,latexsym}
\usepackage{graphicx,amsmath}
\usepackage{color}
\usepackage{bussproofs}
\usepackage[pdfstartview=FitH,pdfpagemode=None,backref,colorlinks=true,citecolor=verydarkblue,linkcolor=verydarkblue]{hyperref}
\usepackage{mathrsfs}
\usepackage{pdfsync}
\usepackage{fullpage}

\let\oldmarginpar\marginpar
\renewcommand\marginpar[1]{\-\oldmarginpar[\raggedleft\footnotesize #1]%
{\raggedright\footnotesize #1}}

\usepackage{mathrsfs} 
\usepackage{enumerate}
\usepackage{stmaryrd}

\renewcommand{\l}{\ell}

\newcommand{\nat}{{\mathbb{N}}}

\newcommand{\rst}{\!\!\upharpoonright\!}

\newcommand{\set}[1]{\left\{#1\right\}}

\newcommand{\F}{\mathbb F}

\newcommand {\cd}{\cdot}
\newcommand {\zo}{\set{0,1}}

\newcommand {\CommaDots} {,\ldots,}
\newcommand {\ML}[1] {\mathbf{M}\!\left[#1\right]} 
\newcommand {\sm} {\setminus}

\newcommand{\Base}{\mbox{}\\ \ind{\textit{Base case: }}}
\newcommand{\Induction}{\mbox{}\\ \ind{\textit{Induction step: }}}
\newcommand{\case}[1]{\ind\textbf{Case #1}:\,}

\newcommand {\ind} {\noindent}

\newcommand {\bigs} {\bigskip}

\DeclareMathAlphabet{\mathitbf}{OML}{cmm}{b}{it}

\font\sf=cmss10
\newcommand{\Nats}{{\hbox{\sf I\kern-.13em\hbox{N}}}}   
\newcommand{\Reals}{{\hbox{\sf I\kern-.14em\hbox{R}}}}  
\newcommand{\Ints}{{\hbox{\sf Z\kern-.43emZ}}}          
\newcommand{\CC}{{\hbox{\sf C\kern -.48emC}}}           
\newcommand{\QQ}{{\hbox{\sf C\kern -.48emQ}}}           

\newcommand{\such}{\;|\;}

\renewcommand{\And}{\land}
\newcommand{\Or}{\lor}
\newcommand{\Not}{\neg}

\newcommand{\BigOr}{\bigvee}

\ifx\theorem\undefined

\newtheorem{theorem}{Theorem}[section]
\newtheorem{lemma}[theorem]{Lemma}
\newtheorem{proposition}[theorem]{Proposition}
\newtheorem{corollary}[theorem]{Corollary}
\newtheorem{definition}{Definition}[section]

\fi

\newtheorem{claim}[theorem]{Claim}

\newenvironment{notation}{\QuadSpace\par\noindent{\bf Notation}:}{\HalfSpace}
\newenvironment{note}{\QuadSpace\par\noindent{\bf Note}:}{\HalfSpace}
\newenvironment{notes}{\QuadSpace\par\noindent{\bf Notes}:}{\HalfSpace}
\newenvironment{importantnote}{\QuadSpace\par\noindent{\bf Important note}:}{\HalfSpace}

\newenvironment{example}{\QuadSpace\par\noindent{\bf Example}:}{\HalfSpace}
\newenvironment{remark}{\HalfSpace\par\noindent{\bf Remark}:}{\HalfSpace}

\ifx\proof\undefined
\newenvironment{proof}{\QuadSpace\par\noindent{\bf Proof}:}{\EndProof\HalfSpace}
\fi

\newenvironment{proofsketch}{\QuadSpace\par\noindent{\textit{Proof sketch}}:}{\EndProof\HalfSpace}

\newenvironment{proofclaim}{\QuadSpace\par\noindent{\bf Proof of claim}:}
{\vrule width 1ex height 1ex depth 0pt $_{\textrm{\,Claim}}$ \HalfSpace}

\newcommand{\QuadSpace}{\vspace{0.25\baselineskip}}
\newcommand{\HalfSpace}{\vspace{0.5\baselineskip}}

\newcommand{\EndProof}{ \hfill \vrule width 1ex height 1ex depth 0pt }

\def\RL0{{\mbox{\rm R(lin)}}}
\def\RZ0{{\mbox{\rm R$^0$(lin)}}}
\def\RC0{R(lin) with constant coefficients}
\def\RCD0#1{{\mbox{\rm R$_{#1}$(lin)}}}

\def\Tse0{{\mbox{$\neg$\textsc{Tseitin}$_{G,p}$}}}

\definecolor{bluetxt}{rgb}{0,0,.5}
\definecolor{myred}{rgb}{0.6,0.0,0.1}
\definecolor{greentxt}{rgb}{0,.5,0}
\definecolor{redtxt}{rgb}{0.1,0.1,0.65}
\definecolor{purpletxt}{rgb}{0.6,0.1,0.7}
\definecolor{black}{rgb}{.0,.0,.0}
\definecolor{verydarkblue}{rgb}{.0,.0,.2}
\definecolor{lightgray}{rgb}{.7,.7,.7}

\ifx\proof\undefined
\newenvironment{proof}{

\smallskip
\noindent\emph{Proof.}}{\hfill\(\Box\)
\bigskip
} \fi

\newlength{\defbaselineskip}
\newcommand{\doublespacing}{\setlength{\baselineskip}{1.0\defbaselineskip}}

\def\ncp0{{\rm NFPC}}
\newcommand{\fop}{{\rm OFPC}}

\title
[Algebraic Proofs over Noncommutative Formulas]{Algebraic Proofs over Noncommutative Formulas}
\author[Iddo Tzameret]{Iddo Tzameret\,$^*$}
\thanks{$^*$Mathematical Institute, Academy of Sciences of the Czech Republic,
\v{Z}itn\'{a} 25, 115 67 Prague 1, Czech Republic. Email: \texttt{tzameret@math.cas.cz}\,.
Supported by The Eduard \v{C}ech Center for Algebra and Geometry and The John Templeton Foundation.}
\date{July 2010}
\address{
    Mathematical Institute, Academy of Sciences
    of the Czech Republic,
\v{Z}itn\'{a} 25, 115 67 Praha 1, Czech Republic.
}

\email{tzameret@math.cas.cz}

\newenvironment{proof_of_example}{\QuadSpace\par\noindent{\textit{Proof of example}}:}{\EndProof\HalfSpace}

\begin{document}
\maketitle
\doublespacing
\begin{abstract}
We study possible formulations of algebraic propositional proof systems operating with noncommutative formulas.
We observe that a simple formulation gives rise to systems at least as strong as Frege---yielding a semantic way to define a Cook-Reckhow (i.e., polynomially verifiable) algebraic analog of Frege proofs, different from that given in \cite{BIKPRS96,GH03}.
We then turn to an apparently weaker system, namely,
polynomial calculus (PC) where polynomials are written as ordered
formulas (\emph{PC over ordered formulas}, for short):
an ordered polynomial is a noncommutative polynomial in which the order of products in
every monomial respects a fixed linear order on variables;
an algebraic formula is \emph{ordered} if the polynomial computed by each of its subformulas is ordered.
We show that PC over ordered formulas is strictly stronger than resolution, polynomial calculus and polynomial calculus with resolution (PCR) and admits polynomial-size refutations for the pigeonhole principle and the Tseitin's formulas.
We conclude by proposing an approach for establishing lower bounds on PC over ordered formulas proofs, and related systems, based on properties of lower bounds on noncommutative formulas.

The motivation behind this work is developing techniques incorporating rank arguments (similar to those used in algebraic circuit complexity) for establishing lower bounds on propositional proofs.
\end{abstract}

\tableofcontents

\section{Introduction}
This work investigates algebraic proof systems establishing propositional tautologies, in which proof lines are written as noncommutative algebraic formulas (noncommutative formulas, for short).
Research into the complexity of algebraic propositional proofs is a central line in proof complexity (cf. \cite{Pit97,Tza08:PhD} for general expositions).
Another prominent line of research is that dedicated to connections between circuit classes and the propositional proofs based on these classes.
In particular, considerable efforts were made to borrow techniques used for lower bounding certain circuit classes, and utilize them to show lower bounds on \emph{proofs} operating with circuits from the given classes.
For example, bounded depth Frege proofs can be viewed as propositional logic operating with $ \textsf{AC}^0 $ circuits, and lower bounds on bounded depth Frege proofs use techniques borrowed from $ \textsf{AC}^0 $ circuits lower bounds (cf. \cite{Ajt88,KPW95,PBI93}).
Pudl{\'a}k et al. \cite{Pud99,AGP01} studied proofs based on monotone circuits---motivated by known exponential lower bounds on monotone circuits.
Raz and the author \cite{RT06,RT07,Tza08:PhD} investigated algebraic proof systems operating with multilinear formulas---motivated by lower bounds on multilinear formulas for the determinant, permanent and other explicit polynomials \cite{Raz04a,Raz04b}.
Atserias et al. \cite{AKV04}, Kraj\'{i}\v{c}ek \cite{Kra07} and Segerlind \cite{Seg07} have considered proofs operating with ordered binary decision diagrams (OBDDs).

The current work is a contribution to this line of research, where the circuit class is noncommutative formulas.
The motivation behind this work is the hope that certain rank arguments, found successful in lower bounding the size of certain algebraic circuits, might facilitate also in establishing lower bounds for the corresponding algebraic proofs.
For this purpose, the choice of noncommutative formulas is natural, since such formulas
constitute a fairly weak circuit class, and the proof of exponential-size lower bounds on noncommutative formulas, given by Nisan \cite{Nis91}, uses an especially transparent rank argument.

We will show that for certain formulations of propositional proof systems over noncommutative formulas demonstrating
lower bounds is likely to be hard,
as the systems we get are considerably strong, and specifically, at least as strong as Frege proofs.
On the other hand, by using a fairly restricted formulation of proofs operating with noncommutative formulas,
we obtain a system that we show is strictly stronger than known algebraic proof systems (like the polynomial calculus).
For this apparently weaker system, demonstrating lower bounds seems not to be outside the reach of current techniques.
In particular, we propose to study the complexity of these proofs by measuring the maximal \emph{rank} of a polynomial appearing in a proof, instead
of the maximal degree (the latter is done in the polynomial calculus).
It is known that the rank of a noncommutative polynomial (as defined for instance by Nisan \cite{Nis91}) is proportional to the minimal size of a noncommutative formula computing the polynomial.
We argue for the usefulness of measuring the maximal rank of a polynomial in algebraic proofs,
by demonstrating a certain property of ranks of ``ordered polynomials'' (as defined formally),
and relating it to proof complexity lower bounds (via an example of a conditional lower bound).
\QuadSpace

\subsection{Results and related works}\label{sec:results}
We concentrate on algebraic proofs establishing propositional contradictions where polynomials are written as noncommutative formulas.
We deal with two kinds of proof systems---both are variants (and extensions) of the polynomial calculus (PC) introduced in \cite{CEI96}.
In PC we start from a set of initial polynomials from $ \F[x_1,\ldots,x_n] $, the ring of polynomials with coefficients from $ \F $ (the intended semantics of a proof-line $ p $ is the equation $ p=0 $ over $ \F $). We derive new proof-lines
by using two basic algebraic inference rules:
from two polynomials $ p$ and $q$, we can deduce $\alpha\cdot
p+\beta\cdot q$, where $\alpha,\beta$ are elements of $\F$; and from $p$ we
can deduce $x_i\cdot p$, for a variable $x_i$ ($i=1,\ldots,n$).
We also have Boolean axioms $x_i^2-x_i=0$, for all $ i=1,\ldots,n $, expressing that
the variables get the values $ 0 $ or $ 1 $.
Our two proof systems extend PC as follows:
\begin{enumerate}
\item PC over noncommutative formulas: NFPC. This proof system operates with noncommutative polynomials over a field, written as (arbitrarily chosen)\footnotemark~noncommutative formulas. The rules of addition and multiplication are similar to PC, except that multiplication is done \emph{either from left or right.}
    We also add a a Boolean axiom $ x_i x_j -x_j x_i $ that expresses the fact that for $ 0,1 $ values to the variables, multiplication is in fact commutative.

\item PC over ordered formulas: OFPC. This proof system is PC operating with ordered polynomials written as (arbitrarily chosen) ordered formulas.
 An ordered polynomial is a noncommutative polynomial such that the order of products in all monomials respects a fixed linear order on the variables, and an ordered formula is a noncommutative formula in which every subformula computes an ordered polynomial.
\end{enumerate}
\footnotetext{This means that if a proof-line consists of the polynomial $ p $, then one may choose to write \emph{any} formula that computes $ p $. (These kind of systems are sometimes called ``semantic'' proof systems.)}

Both proof systems are shown to be Cook-Reckhow systems (that is, \emph{polynomial verifiable}, sound and complete proof systems for propositional tautologies).\QuadSpace

\textbf{(1)} The first proof system NFPC is shown to polynomially simulate Frege (this is partly because of the choice of Boolean axioms).
This gives a semantic definition of a Cook-Reckhow proof system operating with algebraic formulas, simpler in some way from that proposed by Grigoriev and Hirsch \cite{GH03}: the paper \cite{GH03} aims at formulating a formal propositional proof system for establishing propositional tautologies (that is, a Cook-Reckhow proof system), which is an algebraic analog of the Frege proof system.
In order to make their system polynomially-verifiable, the authors augment it with a set of auxiliary rewriting rules, intended to derive algebraic formulas from previous algebraic formulas via the polynomial-ring axioms (that is, associativity, commutativity, distributivity and the zero and unit elements rules). In this framework algebraic formulas are treated as syntactic terms, and one must explicitly apply the polynomial-ring rewrite rules to derive one formula from another.
Our proof system \ncp0 is simpler in the sense that we get a similar proof system to that in \cite{GH03},
while adding no rewriting rules (both our proof system and that in \cite{GH03} can simulate Frege and both are polynomially verifiable and
operate with algebraic formulas, or in our case with noncommutative formulas).
The idea is that because we use noncommutative formulas as proof-lines, to verify that a lines was derived correctly from
previous lines we can use the deterministic polynomial identity testing algorithm for noncommutative formulas devised by Raz and Shpilka \cite{RS04}
(and so we do not need any rewriting rules).
\QuadSpace

\textbf{(2)} For the second proof system, \fop, we show that, despite its apparent weakness, it is stronger than Polynomial Calculus with Resolution (PCR; and hence it is also stronger than both PC and resolution), and also can polynomially simulate a proof system operating with restricted forms of disjunctions of linear equalities called \RZ0 (introduced in \cite{RT07}).
The latter implies polynomial-size refutations for the pigeonhole principle and the Tseitin graph formulas, due to corresponding upper bounds demonstrated in \cite{RT07}.

We then propose a simple lower bound approach for \fop, based on properties of products of ordered formulas (these properties are proved in a similar manner to Nisan's lower bound on noncommutative formulas, by lower bounding the rank of matrices associated with noncommutative polynomials). We show certain sufficient conditions yielding super-polynomial lower bounds on \fop~proofs.

\begin{note}
All the results in this paper hold when one considers \emph{algebraic branching programs} (ABPs) instead of noncommutative formulas,
and \emph{ordered}-ABPs instead of ordered-formulas.
For the precise definition of ABP see e.g., \cite{Nis91}.
An \emph{ordered}-ABP is an ABP such that the order of variables appearing on the edges of
every path from source to sink on the ABP graph, respects a fixed linear order on the variables (see \cite{JQS10} for a close model called
$ \pi $-ordered ABP).
\end{note}

\HalfSpace

\ind\textbf{Related work.}
There is some resemblance between noncommutative formulas (and in fact, algebraic branching programs) and ordered binary decision diagrams (OBDDs) (e.g., close techniques were used to obtain polynomial identity testing algorithms for noncommutative formulas \cite{RS04} and for OBDDs \cite{Waa97}).
Thus, proofs operating with noncommutative formulas are reminiscent to the OBDD-based proof systems introduced in \cite{AKV04,Kra07,Seg07}. Nevertheless, one difference between OBDD-based proofs and noncommutative formulas-based proofs is that the feasible monotone interpolation lower bound technique is applicable in the case of OBDD-based systems, while this technique does not known to lead to super-polynomial size lower bounds even on PC proofs (and thus, also on \fop~proofs which are shown to polynomially simulate PC proofs).

Another proof system, that is even closer to \fop, is that operating with \emph{multilinear formulas} introduced in \cite{RT06} (under the name fMC). The upper bounds on \fop~proofs are similar to that shown for multilinear proofs in \cite{RT06}. Moreover, the technique used by Raz to establish super-polynomial lower bounds on multilinear formulas in \cite{Raz04a} is close---though more involved---to that used by Nisan in the lower bound proof for noncommutative formulas \cite{Nis91}. Therefore, proving lower bounds on \fop~proofs might help in
establishing lower bounds on multilinear proofs as well.

\section{Preliminaries}
For a natural number we let $[n]=\set{1,\ldots,n}$.

\subsection{Noncommutative polynomials and formulas}
Let $ \F $ be a field. Denote by $ \F[x_1,\ldots,x_n] $ the ring of (commutative) polynomials with coefficients from $ \F $ and variables $ x_1,\ldots,x_n $. We denote by $ \F\langle x_1,\ldots,x_n \rangle $ the \emph{noncommutative} ring of polynomials with coefficients from $ \F $ and variables $ x_1,\ldots,x_n $.
In other words, $ \F\langle x_1,\ldots,x_n\rangle $ is the ring of polynomials (where a polynomial is a formal sum of products of variables and field elements) conforming to all the polynomial-ring axioms excluding the commutativity of multiplication axiom. For instance, if $ x_i,x_j $ are two different variables, then $ x_i\cd x_j $ and $ x_j \cd x_i $ are two different polynomials in $\F\langle x_1,\ldots,x_n\rangle $ (note that variables do commute with field elements).

We say that $ \mathcal A $ is an \emph{algebra over $ \F $}, or an \emph{$\F $-algebra}, if $ \mathcal A $ is a vector space over $\F$  together with
a distributive multiplication operation; where multiplication
in $ \mathcal A $  is associative (but it need not be commutative) and there exists a multiplicative unity in $\mathcal A $.

A noncommutative formula is just a (commutative) arithmetic formula, except that we take care for the order in which products are done:

\begin{definition}[Noncommutative formula]\label{def:nonc_formula}
Let $ \F $ be a field and $ x_1,x_2,\ldots $ be variables.
A noncommutative algebraic formula is a labeled tree, with edges directed from the leaves to the root, and with fan-in at most two,
such that there is an order on the edges coming into a node (the first edge is called the \emph{left} edge and the second one the \emph{right} edge). Every leaf of the tree (namely, a node of fan-in zero) is labeled either with an input variable $ x_i $ or a field $ \F $ element.
Every other node of the tree is labeled either with $ +$ or $\times $ (in the first case the node is a plus gate and in the second case a product gate). We assume that there is only one node of out-degree zero, called \emph{the root}.
An algebraic formula \emph{computes} a noncommutative polynomial in the ring of noncommutative polynomials $ \F\langle x_1,\ldots,x_n\rangle  $ in the following way.
A leaf computes the input variable or field element that labels it.
A plus gate computes the sum of polynomials computed by its incoming nodes.
A product gate computes the \emph{noncommutative} product of the polynomials computed by its incoming nodes according to the order of the edges. (Subtraction is obtained using the constant $ -1$.) The output of the formula is the polynomial computed by the root. The depth of a formula is the maximal length of a path from the root to the leaf.
\end{definition}

The \textbf{size} of an algebraic formula (and noncommutative formula) $ f $ is the total number of nodes in its underlying tree, and is denoted $ |f|$.

Raz and Shpilka \cite{RS04} showed that there is a deterministic polynomial identity testing (PIT) algorithm that decides whether two noncommutative formulas compute the same noncommutative polynomial:

\begin{theorem}[PIT for noncommutative formulas \cite{RS04}]\label{thm:RS04-PIT}
There is a deterministic polynomial-time algorithm that decides whether a given noncommutative formula over a field $ \F $ computes the zero polynomial $ 0 $.\footnote{We assume here that the field  $ \F $  can be efficiently represented (e.g., the field of the rationals).}
\end{theorem}

\subsection{Polynomial Calculus}
Algebraic propositional proof systems are proof systems for finite collections of polynomial equations having
no $0,1$ solutions over some fixed field. (Formally, each different field yields a different
algebraic proof system.)
Proof-lines in algebraic proofs (or refutations) consist of polynomials $p$
over the given fixed field.
Each such proof-line is interpreted as the polynomial equation $p=0$.  To consider
the \emph{size} of algebraic refutations we fix the way polynomials inside refutations
are written.

\begin{notation}
An \emph{inference rule} is written as $ \frac{A}{B}\, $ or $ \,\frac{A\ \ B}{C} $, meaning that given the
proof-line $ A $ one can deduce the proof-line $ B $, or given both the proof-lines $ A, B $ one can
deduce the proof-line $ C $, respectively.
\end{notation}

The Polynomial Calculus is a propositional algebraic proof system first considered in \cite{CEI96}:

\begin{definition} \emph{\textbf{(Polynomial Calculus (PC)).}}\label{def:PC}
Let $\mathbb{F}$ be some fixed field and let $Q= \set{Q_1,\ldots,Q_m}$ be a collection of
multivariate polynomials from $\mathbb{F}[x_1,\ldots, x_n]$.
Let the set of axiom polynomials be:
\begin{description}
\item[\quad Boolean axioms]\qquad
$ x_i\cdot(1-x_i)\,, \qquad  \mbox{ for all $\,1\le i\le n$\,.}$
\end{description}
A \emph{PC proof} from  $Q$ of a polynomial $g$ is a finite sequence $\pi =(p_1 ,...,p_\ell)$
of multivariate polynomials from $\mathbb{F}[x_1,\ldots, x_n]$,
where $p_\ell=g$ and for every $ 1\le i\le \l$,
either $p_i = Q_j\,$ for some $j\in[m]$, or $p_i$ is a Boolean axiom, or $p_i$ was deduced from
$p_j,p_k\,$, for  $j,k<i$, by one of the following inference rules:
\begin{description}
  \item[\quad Product]
    \[
        \frac{p}{x_r\cd p}\,, \qquad \,{\mbox{for  $ 1\le r\le n $\,. }}
    \]
  \item[\quad Addition]
    \[
        \frac{{p\quad \quad q}}{{a\cd p + b\cd q}}\,,\qquad \,{\mbox{for}}\  a,b \in \F\,.
    \]
\end{description}
A \emph{PC refutation of} $Q$ is a proof of $\;1$ (which is interpreted as $1=0$, that is the unsatisfiable equation standing for {\rm\textsf{false}}) from $Q$. The \emph{degree} of a PC-proof is the maximal degree of a polynomial in the proof.
The \textbf{size} of a \emph{PC} proof is the total number of monomials (with nonzero coefficients) in all the proof-lines.
\end{definition}

\begin{importantnote}
The size of PC proofs can be defined as the total formula sizes of all proof-lines, where polynomials are written as \emph{sums of monomials}, or more formally, as (unbounded fan-in depth-$ 2 $) $\Sigma\Pi$ formulas.\footnote{A $\Sigma\Pi$ formula $ F $ is an algebraic formula whose underlying tree is of depth $ 2 $ and has unbounded fan-in, such that the root is labeled with a plus gate, the children of the root are labeled with product gates and the leaves are labeled with either variables or field elements.}
This complexity measure is equivalent---up to a factor of $ n $---to the usual complexity measure counting the total number of monomials appearing in the proofs (Definition \ref{def:PC}).
\end{importantnote}

\begin{definition}\emph{\textbf{(Polynomial Calculus with Resolution (PCR)).}}\label{def:PCR}
The PCR proof system is defined similarly to PC (Definition \ref{def:PC}), except that for every variable $ x_i $ a new formal variable $ \bar x_i $ and a new axiom $ x_i +\bar x_i - 1 $ are added to the system, and the Boolean axioms of PCR are as follows:
\begin{description}
\item[\quad Boolean axioms]\qquad
$ x_i\cdot \bar x_i\;.$
\end{description}
The inference rules, and all other definitions are similar to that of PC. Specifically, the \emph{size} of a PCR proof is defined as the total number of monomials in all proof-lines (where now we count monomials in the variables $ x_i $ and $ \bar x_i $).
\end{definition}

\subsection{Proof systems and simulations}
Let $ L\subseteq \Sigma^*$ be a language over some alphabet $\Sigma $.
A \emph{proof system for a language $ L $} is a polynomial-time algorithm $A$ that receives $ x\in\Sigma^* $ and a string $\pi$ over a binary alphabet (``the [proposed] proof" of $x$), such
that there exists a $\pi$ with $A(x,\pi)=\textsf{true}$ if and only if $x\in L $.
Following \cite{CR79}, a \emph{Cook-Reckhow proof system} (or a \emph{propositional proof system}) is a proof system for the language of propositional tautologies in the De Morgan basis $\set{\textsf{true},\textsf{false},\Or,\And,\Not} $ (coded in some efficient [polynomial-time] way, e.g., in the binary $ \set{0,1} $ alphabet).

Assume that $ \mathcal P $ is a proof system for the language $ L $, where $ L $ is not the set of propositional tautologies in De Morgan's basis. In this case we can still consider $ \mathcal P $ as a proof system for propositional tautologies by fixing a translation between $ L $ and the set of propositional tautologies in De Morgan basis (such that $ x\in L $ iff the translation of $ x $ is a propositional tautology [and such that the translation can be done in polynomial-time]).
If two proof systems $ \mathcal P_1 $ and $ \mathcal P_2 $ establish two different languages $ L_1, L_2 $, respectively, then for the task of comparing their relative strength we fix a translation from one language to the other.
In most cases, we shall confine ourselves to proofs establishing propositional tautologies or unsatisfiable CNF formulas.

A propositional proof system is said to be a \emph{propositional refutation system} if it establishes the language of unsatisfiable propositional formulas (this is clearly a propositional proof system by the definition above, since we can translate every unsatisfiable propositional formula into its negation and obtain a tautology).

\begin{definition}\label{defnSim} Let $\mathcal P_1, \mathcal P_2$ be
two proof systems for the same language $ L$ (in case the proof systems are for two different languages we fix a translation from one language to the other, as described above).
We say that $\mathcal P_2$ \emph{polynomially simulates} $\mathcal P_1$
if given a $\mathcal P_1$ proof (or refutation) $\pi$ of a $ F$,
then there exists a proof (respectively, refutation) of $F$ in $\mathcal P_2$ of size polynomial
in the size of $\pi$.
In case $\mathcal P_2$ polynomially simulates $\mathcal P_1$ while $\mathcal P_1$
does not polynomially simulates $\mathcal P_2$ we say that $\mathcal P_2$
is \emph{strictly stronger} than $\mathcal P_1$.
\end{definition}

\section{Polynomial calculus over noncommutative formulas}

\subsection{Discussion}
In this section we propose a possible formulation of algebraic propositional proof systems that operate with noncommutative polynomials. We observe that dealing with \emph{propositional} proofs---that is, proofs whose variables range over $ 0,1 $ values---makes the variables ``semantically'' commutative. Therefore, for the proof systems to be complete (for unsatisfiable collections of noncommutative polynomials over $ 0,1 $ values), one may need to introduce rules or axioms expressing commutativity.
We show that such a natural formulation of proofs operating with noncommutative formulas
polynomially simulate the entire Frege system.

This justifies---if one is interested in concentrating on propositional proof systems weaker than Frege (and especially on lower bounds questions)---our formulation in Section \ref{sec:ordered_ABP_proofs} of algebraic proofs operating with noncommutative algebraic formulas with a \emph{fixed product order} (called \emph{ordered formulas}).
The latter system can be viewed as operating with commutative polynomials over a field precisely like PC, while the complexity of proofs is measured by the total sizes of ordered formulas needed to write the polynomials in the proof. In other words, the role played by the noncommutativity in this system is only in measuring the sizes of proofs: while in PC-proofs the size measure is defined as the number of monomials appearing in the proofs---or equivalently, the total size of formulas in proofs in which formulas are written as (depth-$2$) $\Sigma\Pi$ circuits---the proof system developed in Section \ref{sec:ordered_ABP_proofs} is measured by the total ordered formula size.

\subsection{The proof system NFPC}
We now define a proof system operating with noncommutative polynomials written as noncommutative algebraic formulas.

In algebraic proof systems like the polynomial calculus we transform unsatisfiable propositional formulas into a collection $ Q $ of polynomials having no solution over a field $ \F$.
In the noncommutative setting we translate unsatisfiable propositional  formulas into a collection $ Q $ of noncommutative polynomials from $ \F\langle x_1,\ldots,x_n\rangle $ that have no solution over any noncommutative $\F $-algebra (e.g., the matrix algebra with entries from $ \F $). Although our ``Boolean'' axioms will not force only $ 0,1 $ solutions over noncommutative $\F $-algebras, they will be sufficient for our purpose: every unsatisfiable propositional formula translates (via a standard polynomial translation) into a collection $ Q $ of noncommutative polynomials from $ \F\langle x_1,\ldots,x_n\rangle $, for which $ Q $ and the Boolean axioms have no (common) solution in any noncommutative $ \F $-algebra. Furthermore, \emph{the Boolean axioms will in fact force commutativity of variables product}---as required for variables that range over $ 0,1 $ values (although, again, the Boolean axioms do not force only $ 0,1 $ values when variables range over noncommutative $\F $-algebras).

\begin{definition}[Polynomial calculus over noncommutative formulas: \ncp0]\label{def:ncFPC}
Fix a field $\mathbb{F}$ and let $Q:= \set{q_1,\ldots,q_m}$ be a collection of noncommutative
polynomials from $\F\langle x_1,\ldots, x_n\rangle$.
Let the set of axiom polynomials be:
\begin{description}
\item[\quad Boolean axioms]
\begin{align}\notag
 x_i\cd (1-x_i)\,,              & \qquad  \mbox{ for all $\,1\le i\le n\,$}. \\ \notag
 x_i\cd x_j -x_j\cd x_i \,,     & \qquad  \mbox{ for all $\,1\le i\neq j\le n\,.$}
\end{align}
\end{description}

Let $\pi =(p_1,\ldots,p_\ell)$ be a sequence of noncommutative polynomials from $\F\langle x_1,\ldots, x_n\rangle$, such that for each $i\in[\ell]$, either $\,p_i = q_j\,$ for some
$j\in[m]$, or $p_i$ is a Boolean axiom, or $p_i$ was deduced by one of the
following inference rules using $p_j, p_k\,$, for $j,k<i$:
\begin{description}
  \item[\quad Left/right product]
    \[
        \frac{p}{x_r\cd p}   \qquad \qquad
                    \frac{p}{p\cd x_r} \,,\qquad {\mbox{for  $ 1\le r\le n $\,. }}
    \]
  \item[\quad Addition]
    \[
        \frac{{p\quad \quad q}}{{a\cd p + b\cd q}}\,,\qquad \,{\mbox{for}}\  a,b \in \F\,.
    \]
\end{description}
We say that $ \pi $ is an \emph{\ncp0 proof of $ p_\ell $ from $Q$} if all proof-lines in $ \pi $
are written as noncommutative formulas.
(The semantics of an \emph{\ncp0} proof-line $ p_i$ is the polynomial equation $p_i=0$.)
An \emph{\ncp0 refutation of} $Q$ is a proof of the polynomial $\;1$ from $Q$.
The \textbf{size} of an {\rm \ncp0} proof $\pi$ is defined as the total sizes of all the noncommutative formulas
 in $\pi$ and is denoted $|\pi|$.
\end{definition}

\begin{remark}
(i) The Boolean axioms  might have roots different from $ 0,1 $ over noncommutative $\F $-algebras.
(ii) The Boolean axioms are true for $ 0,1 $ assignments: $ x_i\cd x_j -x_i \cd x_j =0 $ for all $ x_i,x_j \in\zo $.
\end{remark}

We now show that \ncp0 is a sound and complete Cook-Reckhow proof system.
First note that we have defined \ncp0 with no rules expressing the polynomial-ring axioms (the latter are sometimes added to  algebraic proof systems operating with algebraic formulas for the purpose of verifying that every formula in the proof was derived correctly [via the deduction rules of the system] from previous lines; see discussion in Section \ref{sec:results}).
Nevertheless, due to the deterministic polynomial-time PIT procedure
for noncommutative formulas (Theorem \ref{thm:RS04-PIT}) the proof system defined will be a Cook-Reckhow system (that is, verifiable in polynomial-time [whenever the base field and its operations can be efficiently represented]).\QuadSpace

\begin{proposition}
There is a deterministic polynomial-time algorithm that decides whether a given string is an \ncp0-proof (over efficiently represented fields).
\end{proposition}
\begin{proof}
We can assume that the proof also indicates from which previous lines a new line was inferred via the \ncp0  inference rules. Then, by Proposition \ref{thm:RS04-PIT},  there is a polynomial-time algorithm that, e.g.,  given two noncommutative formulas $ F_1,F_2 $ such that the proof indicates that $ F_2 $ was inferred from $ F_1 $ via the Left product rule, decides whether the formula $ x_i\times F_1 $ and $ F_2 $ computes the same noncommutative polynomial.
And similarly for the other deduction rules of \ncp0.
\end{proof}

\begin{proposition}
The systems \ncp0 is sound and complete. Specifically, let $ Q $ be a collection of noncommutative polynomials from $ \F\langle x_1,\ldots,x_n\rangle $. Assume that for every $\F $-algebra, there is no $ 0,1 $ solution for $ Q $ (that is, an $ 0,1 $ assignment to variables that gives all polynomials in $ Q $ the value $ 0 $), then the contradiction $ 1=0 $ can be derived in \ncp0 from $ Q $.
\end{proposition}

\begin{proof}\mbox{}
Soundness holds because both rules of inference are
sound over any $\F $-algebra.
Completeness stems by the simulation of $ \mathcal{F\mbox{\rm-}PC} $ shown in Theorem \ref{theorem:ncp0-simulate-FPC} below (and the fact that if no $\F $-algebra has a solution then also there is no solution in $ \F $ itself, which implies, by completeness of  $ \mathcal{F\mbox{\rm-}PC} $, that there exists an $ \mathcal{F\mbox{\rm-}PC} $  refutation of $ Q $).
\end{proof}

For the next statements we use the algebraic propositional proof system $ \mathcal{F\mbox{\rm-}PC} $ introduced by Grigoriev and Hirsch \cite{GH03} as an algebraic analog of the Frege system.
The proof system $ \mathcal{F\mbox{\rm-}PC} $ is an algebraic propositional proof system operating with (general, that is, commutative) algebraic formulas over a field, and it includes auxiliary rewriting rules allowing to develop equal polynomials syntactically via the polynomial-ring axioms.
The proof system $ \mathcal{F\mbox{\rm-}PC} $ has the Boolean axioms of PC, the rules of PC and in addition the rewrite rules expressing the polynomial-ring axioms. Each line in $ \mathcal{F\mbox{\rm-}PC} $ is treated as a \emph{term}, that is, a formula, and so the rules are also syntactic: addition of terms via the plus gate and product of a term by a variable from the left.
We first need to define the notion of a rewrite rule:

\begin{definition}[Rewrite rule]
A \emph{rewrite rule} is a pair of formulas $ f,g $ denoted $ f \rightarrow g $.
Given a formula $ \Phi $, an \emph{application of a rewrite rule $ f \rightarrow g $ to $\Phi $}
is the result of replacing at most one occurrence of $ f $  in  $ \Phi $  by  $ g $ (that is,
substituting a subformula $ f $  inside $\Phi $ by the formula $ g $).
We write $ f \leftrightarrow g $ to denote the pair of rewriting rules $ f \rightarrow g $ and $ g \rightarrow f $.
\end{definition}

\begin{definition}[$ \mathcal{F\mbox{\rm-}PC} $ \cite{GH03}]\label{def:F-PC}
 Fix a field $ \F $. Let $ F:=\set{f_1,\ldots,f_m} $ be a collection of \emph{formulas}\footnote{Note here that we are talking about formulas (treated as syntactic terms), and \emph{not} polynomials. Also notice that all formulas in $ \mathcal{F\mbox{\rm-}PC} $ are (commutative) formulas computing (commutative) polynomials.}
 computing polynomials from $ \F[x_1,\ldots,x_n] $.
 Let the set of axioms be the following formulas:
 \begin{description}
   \item[\quad Boolean axioms]\qquad
   $ x_i\cdot(1-x_i)\,, \qquad  \mbox{ for all $\,1\le i\le n\,$.}$
 \end{description}
    A sequence $\pi =(\Phi_1,\ldots,\Phi_\ell)$ of formulas computing polynomials from
   $\F[x_1,\ldots, x_n]$\, is said to be \emph{an $ \mathcal{F\mbox{\rm-}PC} $ proof of $ \Phi_\ell $ from $F$},
   if for every $i\in[\ell]$ we have one of the following:
   \begin{enumerate}
      \item $\Phi_i = f_j\,$, for some $j\in[m]$;
      \item $\Phi_i$ is a Boolean axiom;
      \item $\Phi_i$ was deduced by one of the following inference rules from previous proof-lines
            $\Phi_j, \Phi_k\,$, for $j,k<i$:
            \begin{description}
                \item[\quad Product]
                \[
                    \frac{\Phi}{x_r\cd \Phi}\ ,   \qquad \qquad \mbox{for $ r\in[n]$}\,.
                \]
                \item[\quad Addition]
                \[
                    \frac{{\Phi\quad \quad \Theta}}{{a\cd \Phi + b\cd \Theta}}\ ,\qquad
                            \,{\mbox{for}}\  a,b \in \F\,.
                \]
            \end{description}
            (Where $\Phi, x_r\cd\Phi, \Theta, a\cd\Phi, b\cd\Theta $ are \emph{formulas}
            constructed as displayed; e.g.,
            $ x_r\cd\Phi $ is the formula with product gate at the root having the
            formulas $ x_r $ and $ \Phi $
            as children.)\footnote{In \cite{GH03} the product rule of $ \mathcal{F\mbox{\rm-}PC} $ is
            defined so that one can derive $ \Theta\cd\Phi $ from $ \Phi $,
            where $ \Theta $ is any formula, and not just a variable. However, the definition
            of $ \mathcal{F\mbox{\rm-}PC} $ in \cite{GH03} and our Definition \ref{def:F-PC}
            polynomially-simulate each other.}
      \item $\Phi_i $ was deduced from previous proof-line $\Phi_j $, for $ j<i $, by one of the following \emph{rewriting rules}
      expressing the polynomial-ring axioms (where $f,g,h$ range over all algebraic formulas computing polynomials in $\F[x_1,\ldots,x_n]$):
             \begin{description}
                 \item[Zero rule]
                 $ 0\cd f \leftrightarrow 0 $

                \item[Unit rule]
                $  1\cd f \leftrightarrow f$

                \item[Scalar rule]
                $ t \leftrightarrow \alpha $, where $t $ is
                a formula containing no variables (only field $ \F $  elements) that computes the constant $ \alpha\in\F $.

                \item[Commutativity rules]
                $ f + g \leftrightarrow g + f \,$, \qquad $ f\cd g \leftrightarrow g\cd f$
                \item[Associativity rule]
                $ f + (g+h) \leftrightarrow (f+g)+h \,$,  \qquad    $ f\cd(g\cd h) \leftrightarrow (f\cd g)\cd h $
                \item[Distributivity rule]
                $ f \cd(g+h) \leftrightarrow (f\cd g)+(f\cd h) $
             \end{description}

    \end{enumerate}
(The semantics of an $ \mathcal{F\mbox{\rm-}PC} $ proof-line $ p_i$ is the polynomial equation $p_i=0$.)
 An \emph{$ \mathcal{F\mbox{\rm-}PC} $ refutation of} $F$ is a proof of the formula $\;1$ from $F$.
 The \textbf{size} of an  $ \mathcal{F\mbox{\rm-}PC} $ proof $\pi$ is defined as the total sizes of all
 formulas in $\pi$ and is denoted by $|\pi|$.
\end{definition}

\begin{theorem}\label{theorem:ncp0-simulate-FPC}
\ncp0 (over any field) polynomially-simulates Frege. Specifically, \ncp0 polynomially-simulates
$ \mathcal{F\mbox{\rm-}PC} $ in the following sense: let $ f_1,\ldots,f_m $ be a set of commutative formulas computing (commutative) polynomials that have no common $ 0,1 $ root, and assume that there is a size $ s $\, $ \mathcal{F\mbox{\rm-}PC} $ refutation of $ f_1,\ldots,f_m $. Then,  there exists an \ncp0\ refutation of the same set of formulas $ f_1,\ldots,f_m $ (but now viewed as computing noncommutative polynomials) of size polynomial in $ s $.
\end{theorem}

\begin{proof}
By \cite{GH03} (see Theorem 3 there), $ \mathcal{F\mbox{\rm-}PC} $ polynomially simulates Frege.
We proceed by showing a simulation of $ \mathcal{F\mbox{\rm-}PC} $ by \ncp0 by induction on the number of steps in an $ \mathcal{F\mbox{\rm-}PC} $ proof.

\Base Axioms and initial formulas. All axioms of $ \mathcal{F\mbox{\rm-}PC} $ are also axioms in \ncp0. Also, if the $ \mathcal{F\mbox{\rm-}PC} $ refutation uses an initial formula $ f_i $, then we use the same formula in \ncp0.

\Induction

\case 1 Addition rule.
Assume we derive in $ \mathcal{F\mbox{\rm-}PC} $ the \emph{formula} $ p+q $.
By induction hypothesis we already have the two formulas $ p,q$ in \ncp0.
Thus, we can add them via the addition rule.

\case 2 Product rule.
Assume we derive the formula $ x_i\cd p$ from the formula $ p $ in $ \mathcal{F\mbox{\rm-}PC} $.
By induction hypothesis we already have the formula $ p $ in \ncp0.
Thus, we can derive $ x_i\cd p $ by the Left product rule.

\case 3 Rewriting rules. Assume we derived a formula $ f $ using one of the rewriting rules of $ \mathcal{F\mbox{\rm-}PC} $.
The rewriting rules of associativity, distributivity, scalar rule, and unit and zero rules of $ \mathcal{F\mbox{\rm-}PC} $ \emph{do not change the noncommutative polynomial computed by an algebraic formula}.
Therefore, we get them ``for free'' in \ncp0, in the sense that
we can choose to write a noncommutative polynomial $ p $ in the proof as any noncommutative formula, as long as the chosen formula computes the noncommutative polynomial $ p $.
Thus, we only need to show how to simulate the commutativity rule, namely to show how to simulate commuting a term inside a formula. The key lemma for this is the following:

\begin{lemma}\label{lem:simulate_gf_fg}
Let $ \F $ be any field and let $ f,g $ be two noncommutative formulas computing (non-constant) polynomials from $ \F\langle x_1,\ldots,x_n\rangle $. Then, there is an \ncp0 proof of size polynomial in $ |f| + |g| $ of the formula $ f\cd g - g\cd f $.
\end{lemma}

\begin{proof}
First, we need to show that \ncp0 allows for substitution of identities inside proof-lines.
Let $ A,h $ be noncommutative formulas and assume that the variable $ z $ occurs inside $ A $ only once. Then  $ A[h/z] $ denotes the noncommutative formula obtained from $ A $ by replacing the leaf labeled $ z $ by the formula $ h $.

\begin{claim}\label{cla:substitution_in_ncp0}
Let $ A $ be a noncommutative formula, and let $ z $ be a variable that occurs only once inside $ A $. Let $ h,h' $ be two noncommutative formulas $ h,h'$ of maximal size $s $. Then, there is an \ncp0 proof of  $ A[h/z] - A[h'/z] $ from $ h-h' $ of size polynomial in $ |A|+s $.
\end{claim}

\begin{proofclaim}
Straightforward induction on the size of $ A$.
\end{proofclaim}

We get back to the proof of Lemma \ref{lem:simulate_gf_fg}:
proceed by induction on $ |f|+|g| \ge 2 $.
\Base $ |f|+|g| =2 $. By assumption the polynomials computed by $ f,g $ are both non-constant, and so  $ f=x_i $ and $ g=x_j $, for some $ i,j\in[n]$. Therefore, we are done by the Boolean axiom $ x_i x_j - x_j x_i $ .

\Induction Either $ |f|>1 $ or $ |g|>1 $. Assume without loss of generality that $ |f|>1 $.
Following Claim \ref{cla:substitution_in_ncp0}, we shall use freely substitutions in formulas.
\QuadSpace

\case {(i)} $ f=f_1+f_2 $. Start from
\begin{equation}\label{eq:1000}
f\cd g -f\cd g =f\cd g -(f_1+f_2)\cd g = f\cd g -f_1\cd g -f_2\cd g\,.
\end{equation}
By induction hypothesis we have a proof of $ f_1\cd g - g\cd f_1 $ and of $ f_2\cd g- g\cd f_2 $. Thus, we can substitute these identities in (\ref{eq:1000}), to get $ f\cd g - g\cd f_1 -g\cd f_2 =
f\cd g - g\cd(f_1 +f_2)=f\cd g - g\cd f$.

\case {(ii)}  $ f=f_1\cd f_2 $. Start from
\begin{equation}\label{eq:1010}
f\cd g -f\cd g =f\cd g -(f_1\cd f_2)\cd g = f\cd g -f_1\cd (f_2\cd g)\,.
\end{equation}
By induction hypothesis we have a proof of $ f_2\cd g - g\cd f_2 $. Thus, we can substitute this identity in (\ref{eq:1010}), to get $ f\cd g - f_1\cd(g\cd f_2) = f\cd g - (f_1\cd g)\cd f_2$. By induction hypothesis again, we have $ f_1\cd g- g\cd f_1 $. And similarly, we get by substitution
$ f\cd g - (g\cd f_1)\cd f_2=f\cd g - g\cd f$.

This concludes the proof of Lemma \ref{lem:simulate_gf_fg}
\end{proof}

To conclude the simulation of the commutativity rewrite rule of $ \mathcal{F\mbox{\rm-}PC} $ (which will also conclude the proof of Theorem \ref{theorem:ncp0-simulate-FPC}) we notice that, by Claim \ref{cla:substitution_in_ncp0} and by Lemma \ref{lem:simulate_gf_fg},
for any noncommutative formula $ A $, such that $ z $ is a variable that occurs only once inside $ A $, there is an \ncp0 proof of  $ A[(f\cd g)/z] - A[(g\cd f)/z] $ of size polynomial in $ \left|A[(f\cd g)/z]\right| $.
\end{proof}

\section{Polynomial calculus over ordered formulas}
\label{sec:ordered_ABP_proofs}
In this section we formulate an algebraic proof system \fop~that operates with noncommutative polynomials  in which every monomial is a product of variables in \emph{nondecreasing order} (from left to right; and according to some fixed linear order on the variables), and where polynomials in proofs are written as \emph{ordered formulas}, as defined below.

Let $ X=\set{x_1,\ldots,x_n} $ be a set of variables and let $\F $ be a field. Let $ \preceq $ be a linear order on the variables $ X $. Let $ f=\sum_{j\in J} b_j \mathcal M_j $ be a commutative polynomial from $  \F[x_1,\ldots,x_n] $, where the  $ b_j $'s are coefficient from $ \F $ and the $ \mathcal M_j $'s are monomials in the $ X $ variables.
We define $\llbracket f \rrbracket  \in \F\langle x_1,\ldots,x_n\rangle $ to be the (unique) noncommutative polynomial $ \sum_{j\in J} b_j \cd \llbracket \mathcal M_j \rrbracket$, where $ \llbracket \mathcal M_j\rrbracket $ is the (noncommutative) product of all the variables in $ \mathcal M_j $ such that the order of multiplications respects $ \preceq $.
We denote the image of the map $ \llbracket \cd \rrbracket :\F[x_1,\ldots,x_n] \to \F\langle x_1,\ldots,x_n\rangle $ by $ \mathcal G $. We say that a polynomial is an \emph{ordered polynomial} if it is a polynomial from $ \mathcal G $.

\begin{definition}[Ordered formula]
\label{def:ord-product-formula}
Let $ \preceq $ be some fixed linear order on variables $ x_1,\ldots,x_n $.
A noncommutative formula (Definition \ref{def:nonc_formula}) is said to be an \emph{ordered formula} if the polynomial computed by each of its subformulas is ordered.
We say that an ordered formula $ F $ computes the \emph{commutative} polynomial $ f\in \F[x_1,\ldots,x_n] $ whenever $ F $ computes $ \llbracket f\rrbracket $.
\end{definition}

An equivalent characterization of ordered formulas is as \emph{syntactic ordered formulas}:

\begin{definition}[Syntactic ordered formula]
An ordered formula is a \emph{syntactic ordered formula} if for each of its product gates the left subformula contains only variables that are less-than or equal, via $ \preceq $, than the variables in the right subformula of the gate.
\end{definition}

\begin{proposition}\label{prop:syntactic-ordered-formulas-equal-ordered-fmlas}
There is a polytime algorithm that receives an algebraic formula $\Phi $ and a linear order on its variables,
and returns {\rm \textsf{false}} if $ \Phi $ is not an ordered formula, and otherwise returns a syntactic ordered formula of
the same size as $ \Phi $ that computes the same (ordered) polynomial.
\end{proposition}
\begin{proof}
The algorithm is as follows:
Search for a product node in $ F $ that has on its left subformula a variable that is greater (via the order $ \preceq $) than some variable in its right subformula. If there is no such product node, then $ F $ itself is a syntactic ordered formula, and the algorithm returns $ F $.

Otherwise, let $ v $ be a product gate in $ F $, with $ F_1 $ and $ F_2 $ its left and right subformulas, respectively.
And suppose that $ F_1 $ contains the variable $ x_i $ and $ F_2 $ contains the variable $ x_j $, such that $ x_i\succ x_j $.
Let $ h_1,h_2 $ be the polynomials computed by $ F_1 $ and $ F_2$, respectively.

We first check whether $ x_i $ occurs in $ h_1 $.
To this end we substitute every appearance of $ x_i $ in $ F_1 $ by the constant $ 0 $,
and check if the resulted formula, denoted $ F_1(0/x_i) $,  computes the same noncommutative polynomial as $ F_1 $ (using the PIT algorithm for noncommutative formulas).
If the answer to the latter question is ``\textsf{yes}", then we conclude that $ x_i $ does not occur in the polynomial $ h_1 $,
and we run the algorithm with the input formula $ F $ in which $ F_1 $ is substituted by $ F_1(0/x_i) $.
If the answer to the question was ``\textsf{no}", we check in a similar manner whether $ x_j $ occurs in $ h_2 $.
If $ x_j $ does not occur in $ h_2 $ we run the algorithm with the formula $ F $ in which $ F_2 $ is substituted by $ F_2(0/x_j) $
(where $ F_2(0/x_j) $ is $ F_2 $ after substituting every appearance of $ x_j $ by $ 0 $).
If $ x_j $ \emph{does} occur in the polynomial $ h_2 $,
then the polynomial computed at $ v $ is not ordered (since we already know that $ x_i $ occurs in $h_1 $, and so $ h_1\cd h_2 $ is not an ordered polynomial), and so $ F $ is not an ordered formula, and we return \textsf{false}.\HalfSpace

Note that the algorithm described above returns either \textsf{false} (in case $ F $ is not an ordered formula) or
a new formula that computes the same (noncommutative) polynomial as $ F $ and with \emph{the same size} as $ F $ (because the only changes applied to the original formula $ F $ is substitution of variables by the constant $ 0 $). The running time of the algorithm is polynomial in the size of $ F $.
\end{proof}

We can now define \fop~in a convenient way, that is, without referring to noncommutative polynomials: the system \fop~is defined similarly to PC, except that the proof-lines are written as ordered formulas:

\begin{definition}[PC over ordered formulas (\fop)]\label{def:ncFPC}
Let $ \pi=(p_1,\ldots,p_m) $ be a PC proof of $ p_m $ from some set of initial polynomials $ Q $ (that is, $ p_i $ are commutative polynomials from the ring of polynomials $ \F[x_1,\ldots,x_n]  $), and let $ \preceq $ be some linear order on the variables $ x_1,\ldots,x_n  $. The sequence $(f_1,\ldots,f_m)$ in which $ f_i $ is an ordered formula computing $ p_i $ (according to the order $ \preceq $),
is called an \fop~proof of $ p_m $ from $ Q $.
The \textbf{size} of an \fop~proof is the total sizes of all the ordered formulas appearing in it.
\end{definition}


Similar to the proof system \ncp0 we have defined \fop\ with no rules expressing the polynomial-ring axioms.
Also, similar to \ncp0,
the system \fop~will constitute a Cook-Reckhow proof system, that is, there is a deterministic polynomial-time algorithm that decides whether a given string is an \fop~proof or not (whenever the base field and its operations can be efficiently represented):

\begin{proposition}
For any linear order on the variables, \fop~is a sound, complete and polynomially-verifiable refutation system for establishing  that a collection of
(commutative) polynomial equations over a field does not have $ 0,1 $ solutions. Specifically, (considering the language of polynomial translations of Boolean contradictions) \fop~is a Cook-Reckhow proof system.
\end{proposition}
\begin{proof}
The soundness and completeness of \fop~stem from the soundness and completeness of PC. The fact that \fop~is a Cook-Reckhow proof system is proved in Proposition \ref{prop:OFPC_is_Cook_Reckhow} below.
\end{proof}

We first need the following lemma:
\begin{lemma}\label{lem:pit-ordered-formulas}
For any linear order $ \preceq $ on variables, there exists a polytime algorithm that receives an ordered formula $ \Phi $ computing $ \llbracket f \rrbracket \in \F\langle x_1,\ldots,x_n\rangle $ (for some polynomial $ f \in \F[x_1,\ldots,x_n] $) and a variable $ x_r $, for some $ 1\le r\le n $, and outputs a new ordered formula that computes $ \llbracket x_r\cd f \rrbracket $.
\end{lemma}

\begin{proof}
We can assume that $ \Phi $ is a \emph{syntactic} ordered formula, as otherwise we can transform it into such a formula by using the algorithm in Proposition \ref{prop:syntactic-ordered-formulas-equal-ordered-fmlas}.
By induction on the size of the formula $ \Phi $, we show that there is an algorithm $ A(\Phi,x_r) $ that outputs
the correct formula.
\vspace{-9pt}

\Base
\begin{enumerate}
\item $A(c,x_r):=c\cd x_r $, for $ c\in \F $.
\item $ A(x_i,x_r):=x_r\cd x_i $ or $ A(x_i,x_r):=x_i\cd x_r $, depending on whether
$ x_i\preceq x_r $ or $ x_i\preceq x_r $, respectively.
\end{enumerate}
\vspace{-11pt}

\Induction
\begin{enumerate}
\item $ A(\Phi_1+\Phi_2,x_r):=A(\Phi_1,x_r)+A(\Phi_2,x_r) $.
\item $ A(\Phi_1\cd\Phi_2, x_r):= A(\Phi_1, x_r)\cd\Phi_2 $, in case $ x_i $ is less-than or equal ($\preceq $) than every variable
in $ \Phi_2 $, and otherwise $ A(\Phi_1\cd\Phi_2, x_r):= \Phi_1\cd A(\Phi_2, x_r) $.
\end{enumerate}
\end{proof}

\begin{proposition}\label{prop:OFPC_is_Cook_Reckhow}
For any linear order $ \preceq $ on variables, there exists a polytime algorithm
that given a sequence $ \pi $ of ordered formulas and another sequence  $ Q_1,\ldots,Q_m,g $ of ordered formulas,
outputs $ 1 $ iff  $ \pi $ is an \fop\ proof of the polynomial computed by $ g $ from the polynomials computed by $ Q_1,\ldots,Q_m $.
\end{proposition}

\begin{proof}
 We verify the following:
 \begin{enumerate}
       \item All formulas in $ \pi $ are ordered formulas (according to the fixed linear order). By Proposition \ref{prop:syntactic-ordered-formulas-equal-ordered-fmlas}, this can be done in polynomial-time in the size of $ \pi $.
       \item The last formula in $ \pi $ computes $ g $. This can be done by checking that
            the last formula in $ \pi $ computes the same noncommutative polynomial as $ g $ (using the PIT algorithm for noncommutative formulas in Theorem \ref{thm:RS04-PIT}).
       \item For every proof-line $ f\in\pi$ one of the following holds:
       \begin{enumerate}
             \item[(i)] The formula $ f$ computes an axiom. This can be verified by checking whether
             $ f $ computes the same noncommutative polynomial as the formula $ x_i^2-x_i $, for some $ 1\le i\le n $, or whether $ f $ computes some polynomial computed by $ Q_i $, for some $ 1\le i\le m $ (again, by the PIT algorithm for
             noncommutative formulas).
             \item[(ii)] The formula $ f $ computes the same ordered polynomial as $ F+G $, for some pair $ F,G $ of ordered formulas in previous proof-lines (verify by the PIT algorithm for noncommutative formulas).
             \item[(iii)] The formula $ f $ computes $ \llbracket x_i\cd h \rrbracket $, for some $ 1\le i\le n $, where $ h $ is a polynomial computed by some previous proof-line. To check this we do the following: considering a previous proof-line $ H $, we run the algorithm in Lemma \ref{lem:pit-ordered-formulas} where the inputs are $ H $ and $ x_i $. We get a new
                 ordered formula $ H' $, and we check if $ H' $ computes the same noncommutative polynomial as $ f $.
       \end{enumerate}
 \end{enumerate}
\end{proof}

\begin{notes}
\begin{enumerate}
\item In case we assume that there is an apriori fixed linear order of variables, we may speak about ordered formulas without referring explicitly to some linear order.
\item Formally, for different $ n $'s, every set of variables $ x_1,\ldots,x_n $ may have linear orders that are incompatible with each other. Nevertheless, in this paper, given a family $ Q $ of collections of initial polynomials $ \set{Q_n\,|\, n\in\nat}$ parameterized by $ n $, and assuming that $ Q_n\subseteq \F[x_1,\ldots,x_n] $ for all $ n $, we will consider only linear orders such that: for every $ n>1 $, the linear order on $ x_1,\ldots,x_n $ is an \emph{extension} of the linear order on $ x_1,\ldots,x_{n-1} $. Equivalently, we can consider one fixed linear order on a countable set of variables $ X=\{x_1,x_2,\ldots\} $.
\end{enumerate}
\end{notes}

\section{Simulations, short proofs and separations for \fop}
In this section we are concerned with the relative strength of \fop. ~Specifically, we show that \fop~is strictly stronger than the polynomial calculus, polynomial calculus with resolution (PCR, for short; see Definition \ref{def:PCR}) and resolution (for a definition, see for example \cite{ABSRW99}).
For this purpose, we show first that, for any linear order on the variables, \fop~polynomially simulates PCR. Since PCR polynomially simulates both PC and resolution, we get that \fop~ also polynomially simulates PC and resolution.
Second, we show that \fop~admits polynomial-size refutations of tautologies (formally, families of unsatisfiable collections of polynomial equations) that are hard (that is, do not have polynomial-size proofs) in PCR.

Let $ \tau $ denote the linear transformation that maps the variables $ \bar x_i $, for any $ i\in[n] $, to $ (1-x_i) $, and denote $ p\rst\tau $ the polynomial $ p $ under the transformation $ \tau $.
\begin{proposition}
For any linear order on the variables, \fop~polynomially simulates PCR (and PC and resolution).
Specifically, if there is a size $ s $ PCR proof (with the variables $ x_1,\ldots,x_n, \bar x_1,\ldots,\bar x_n $) of $ p $ from the axioms $ p_{j_1},\ldots,p_{j_k} $, then there is an \fop~proof of $ p\rst \tau $ from $ p_{j_1}\rst\tau,\ldots,p_{j_k}\rst\tau $ of size $ O(n\cd s) $.
\end{proposition}
\begin{proof}
Given some linear order on the variables, we assume that all ordered formulas respect this linear order (and so we do not refer explicitly to this order).

Let $ \pi=(p_1,\ldots,p_t) $ be a PCR proof of size $ s $ from the axioms $ p_{j_1},\ldots,p_{j_k} $ (that is, $ p_i $'s are [commutative] polynomials from $ \F[x_1,\ldots,x_n, \bar x_1,\ldots,\bar x_n] $, for some field $ \F $, such that the total number of monomials occurring in all proof-lines in $\pi $ is $ s $). We need to show that there is an \fop~proof $\pi'$ of $ p_i $ from the axioms, such that $\pi' $ has size $ O(n\cd s) $.

Let $ \Gamma $ be the sequence obtained from $ \pi $ by replacing every product rule application in $ \pi $, deriving $ \bar x_i \cd p $ from $ p $ (for any $ i=1,\ldots,n $), by the following proof sequence:
\begin{enumerate}[\ \ 1.]
\item $ p $
\item $ x_i \cd p $
\item $ (1-x_i)\cd p $
\end{enumerate}
(the second polynomial is derived by the product rule from the first polynomial, and the third polynomial is derived by the addition rule from the first and second polynomials).

Let $ \Gamma\rst\tau$ be the sequence obtained from $ \Gamma $ by applying the substitution $\tau $ on every proof-line in $ \Gamma $.
We claim that $\Gamma\rst\tau$ is a PC proof of $ p_t\rst\tau$ from the initial polynomials $ p_{j_1}\rst\tau,\ldots,p_{j_k}\rst\tau$:
first, note that all product rule applications using $ \bar x_i $ variables were eliminated in $ \Gamma\rst\tau$, and thus all product rule applications in $ \Gamma\rst\tau $ are legitimate PC product rule applications. Second, note that for any pair of polynomials $ g,h $ we have $ g\rst\tau+h\rst\tau = (g+h)\rst\tau $. Third, note that the axioms of PCR transform under $ \tau $ to either $ 0 $ (which we can ignore in the new proof sequence) or to the PC axiom $ x_i(1-x_i) $.

By construction, every proof-line in $ \Gamma\rst\tau $ is either $ p_i\rst\tau $ or $ x_j\cd (p_i\rst\tau) $, for some $ p_i\in\pi $ and $ j\in[n] $. Therefore, by definition of \fop, it suffices to show
that every $ p_i\rst\tau $ and $ x_j\cd (p_i\rst\tau) $, for some $ p_i\in\pi $ and $ j\in[n] $, have ordered formulas of size at most $ O(m\cd n) $, where $ m $ is the number of monomials in $ p_i $.
For this purpose it is enough to show that for every monomial $ \mathscr M $ in $ p_i $ there exists an $ O(n) $ ordered formula computing the polynomial $ \mathscr  M\rst\tau $.
The latter is true since every such polynomial is a product of at most $ n $ terms, where each term is either $ x_i $ or $ 1-x_i $, for some $ i\in[n] $; such a product can be clearly written as an ordered formula of size $ O(n)$.
\end{proof}

\subsubsection{\fop~polynomially simulates \RZ0}\label{sec:fop_sim_RZ0}
We now show that \fop~can polynomially simulate the proof system \RZ0 introduced in \cite{RT07}. This will be used in Section \ref{sec:cor_for_fop} to establish the \fop~upper bounds.
In that paper a refutation system \RL0 was introduced. \RL0 is a refutation system extending resolution to work with disjunctions of linear equations instead of disjunction of literals.
\RZ0 is defined to be a subsystem of \RL0 in which certain restrictions put on the possible
disjunctions of linear equations allowed in a proof. For the precise definition of \RL0 and \RZ0 we refer the reader to \cite{RT07}. However, it is not entirely necessary to know the definitions of \RL0 and \RZ0, since we will use a polynomial translation of \RZ0 defined below, and describe explicitly what is needed for the proofs ahead.

First, we need the definitions that follow.
A \emph{polynomial translation of a clause $ \BigOr_{j\in J}(x_j^{b_j}) $ } is a any product of the form $ \prod_{j\in J}(x_j-b_j) $, where $ b_j\in\zo $ for all $ j\in J $, and where $ x_j^{b_j} $ is the literal $ x_j $ if $ b_j=1 $ and $ \neg x_j $ if $ b_j=0 $. Accordingly, we define the \emph{polynomial translation of a CNF formula} as the set consisting of the polynomial translations of the clauses in a CNF.

\begin{definition}[Polynomial translation of \RCD0{c,d}-lines]\label{def-poly-trnas-of-R(lin)-lines}
A \emph{polynomial translation of an \RCD0{c,d}-line} is a product $ D=\prod_{j\in J} L_j $, where the $ L_j $'s are linear forms, and:
\begin{enumerate}
\item All variables in the linear forms have integer coefficients with absolute values at most $c$ (the constant terms are unbounded).

\item $ D $ can be written as $ \prod_{i=1}^d D_i$, where each $D_i$ either consists of (an unbounded) product of linear forms \emph{that
differ only in their constant terms}, or is a translation of a clause (as defined above).
\end{enumerate}
The \emph{width} of a polynomial-translation of an \RCD0{c,d}-line $D$ is defined to be the
total degree of the polynomial $ D $.
\end{definition}

In other words, any polynomial translation of an \RCD0{c,d}-line has the following general form:

\begin{equation}\label{eq-R0-clause}
    \prod_{j\in J}(x_j-b_j)  \cd
     \prod_{t=1}^k
    \prod_{i\in I_t}
            \left(
             \sum_{r=1}^n a_r^{(t)} x_r -\ell^{(t)}_i
        \right)
\,,
\end{equation}
where $k\le d$  and  for all $r\in[n]$ and $t\in[k]$, $a^{(t)}_r$ is
an integer such that $|a^{(t)}_r|\le c$,
and $b_j\in\zo$ (for all $j\in J$) (and $I_1\CommaDots I_k, J$ are
unbounded sets of indices).
Clearly, a disjunction of clauses is a clause in itself, and so
we can assume that in any \RCD0{c,d}-line
only a single polynomial translation of a clause occurs.

We shall use the following propositions:

\begin{proposition}[Algebraic translation of \RZ0; Corollary 9.11 \cite{RT07} (restated)]
\label{cor-PCR-RZ0-lines-translation}
Let $K:= \set{K_n\such n\in\nat}$ be a family of unsatisfiable CNF formulas\footnote{Formally, we have a straightforward translation of CNFs to the language of \RZ0 (see \cite{RT07}).},
and let $\set{P_n\such n\in\nat}$ be a family of \emph{\RZ0-proofs} of $K$.
Then, there are two constants $c,d$ that do not depend on $ n $ and
a family of PC proofs $\set{P'_n\such n\in\nat} $ of the polynomial translations of the family of CNFs $ K $, such that for every $ n $ the proof $ P'_n $ has polynomial-size in the size of $P_n$ number of steps, and where every line in $P'_n$ is a (polynomial translation of an) \RCD0{c,d}-line (Definition \ref{def-poly-trnas-of-R(lin)-lines}) whose width is polynomial in the size of $ P_n $.
\end{proposition}

\begin{remark}
It is immaterial to define the size measure for \RZ0 refutations (though this concept is mentioned in Theorem \ref{cor-PCR-RZ0-lines-translation}); we shall only use the fact that \RZ0 has short refutations for some hard contradictions.
\end{remark}

\begin{note}
Although corollary 9.11 in \cite{RT07} is stated for PCR instead of PC, the translation holds also for PC
(see Remark before Corollary 9.11 in \cite{RT07}).
\end{note}

\begin{definition}[Multilinearization operator]\label{def-ml-operator}
Given a field $\;\mathbb{F}$ and a polynomial $q\in
\mathbb{F}[x_1,\ldots,x_n]$, we denote by $\ML{q}$ the unique multilinear
polynomial equal to $q$ modulo the ideal generated by all the polynomials
$\,x_i^2-x_i$, for all variables $x_i$.
\end{definition}

For example, if $q=x_1^2x_2+a x_4^3\,$ (for some $a\in \mathbb{F}$)
then $\,\ML{q}=x_1x_2+a x_4\,$. \bigs

\begin{proposition}[Implicit in \cite{RT06,RT07}]\label{prop:implicit_in}
Let $ P $  be a PCR refutation from initial multilinear polynomials. Then we can transform $ P$  into a new PCR refutation $ P' $ from the same initial multilinear polynomials such that $ P' $ contains only multilinear polynomials, with only a polynomial increase in the number of steps. Moreover, if the proof lines in $ P $ are all \RCD0{c,d}-lines of maximal width $ w $, then all the proof lines in $ P'$  are multilinearizations of \RCD0{c',d'}-lines of maximal  width polynomial in $ w $ and where $ c',d' $ depend only on $ c,d $.
\end{proposition}

\begin{proofsketch}
Given a PCR proof $ P =(p_1,\ldots,p_m)$ in the variables $ \set{x_1,\ldots,x_n,\bar x_1,\ldots,\bar x_n} $, consider the sequence $ S $ of multilinearized polynomials $ (\ML{p_1},\ldots,\ML{p_m}) $.
Then, by the proof of Theorem 5.1 in \cite{RT06} one can add polynomially in $ m $ many multilinear polynomials to $ S $ so that the new sequence $ S' $ consists of only multilinear polynomials and constitutes a PCR refutation of the initial polynomials. (Theorem 5.1 from \cite{RT06} talks about fMC refutations [Definition 2.6 in \cite{RT06}]. However, it is clear from the definition of fMC that the underlying sequence of polynomials in any fMC refutation constitutes a PCR refutation as well.)

Assume in addition that all polynomials in $ P $ are polynomial translations of \RCD0{c,d}-lines (Definition \ref{def-poly-trnas-of-R(lin)-lines}). Then, $ S =(\ML{p_1},\ldots,\ML{p_m}) $ is a sequence of multilinearizations of \RCD0{c,d}-lines. The only thing left to check is that the additional polynomials added to $ S $ to yield $ S' $ in the proof of Theorem 5.1 \cite{RT06} \emph{are all polynomial translations of \RCD0{c',d'}-lines}, where $ c',d' $ depend only on $ c,d $.
This could be done by straightforward inspection of the proof of Theorem 5.1 \cite{RT06}.
\end{proofsketch}

Now we are ready to prove the main simulation of this subsection:

\begin{theorem}\label{thm:FOP_sim_RZ0}
For any linear order on the variables, \fop~polynomially simulates \RZ0 (over large enough fields). Moreover, we can assume that all formulas appearing in the \fop~proofs simulating \RZ0 are depth-$3$ ordered formulas.
\end{theorem}

By Propositions \ref{cor-PCR-RZ0-lines-translation} and \ref{prop:implicit_in} and by the definition of \fop, in order
to prove Theorem \ref{thm:FOP_sim_RZ0}
it suffices to prove the following lemma (implicit in \cite{RT07}):

\begin{lemma}[Implicit in Lemma 9.14 \cite{RT07}]\label{lem:ordered_formula_for_RCD-line}
Let $ p $ be a polynomial translation of an \RCD0{c,d}-line of width $ w $ over $ n $ variables.
Then, $\ML{p} $ can be computed by an ordered formula of size polynomial in $ w \cd n $
over fields of size bigger than $\, w\cd n $. Moreover, the ordered formula is a $\Sigma\Pi\Sigma$ formula\footnote{This means that every path from the root to the leaf in the formula tree starts with a plus gate, and the number of alternation in the path between plus and product gates is at most two}.
\end{lemma}

\begin{proof}
The proof uses the fact that \RCD0{c,d}-lines are close to a product of $ d $ symmetric polynomials, and the fact that symmetric polynomials can be computed by small ordered formulas (of depth-$3$) over large enough fields. Specifically:

\begin{claim}[Restatement of Claim 9.15 in \cite{RT07}]\label{cla-tedious-summands}
Let $ D $ be a polynomial translation of an \RCD0{c,d}-line of width $ w $.
Then, $ D$ is a linear combination (over $\F$) of $(w+c)^{c\cd d}$ many terms,
such that each term is of degree at most $ w $ and can be written as
\begin{equation}
\label{eq:the_zk_terms}
    q\cd\prod_{k\in K} z_k^{r_k}\,,
\end{equation}
where $K$ is a collection of indices such that $ |K|\le c\cd d  $,
and $r_k$'s are non-negative integers $\le w $,
and the $z_k$'s are homogenous linear forms such that each $z_k$ has a single integral coefficient for all variables
in it\footnote{That is, $ z_k = b\cd\sum_{j=1}^l x_{i_j} $ for some natural number $ b $.}, and  $q$
is a polynomial translation of a clause.
\end{claim}

By this claim, to complete the proof of Lemma \ref{lem:ordered_formula_for_RCD-line}
it is sufficient to show that the multilinearization of any term as in (\ref{eq:the_zk_terms}):
\begin{equation}\label{eq:multilinear_of_zk_products}
\ML{q\cd\prod_{k\in K} z_k^{r_k}}\,
\end{equation}
can be computed by an ordered $\Sigma\Pi\Sigma$ formula of size polynomial in $ cdn $, over fields of size bigger than $c\cd w $.
This is done by using polynomial interpolation, as shown (implicitly) in Claim 9.16 in \cite{RT07}. More specifically,  Claim 9.16 in \cite{RT07} demonstrated that (\ref{eq:multilinear_of_zk_products}) can be computed by a formula $ \Phi $ such that:
(i) $\Phi $ consists of polynomially in $ d, c $ many summands; (ii) each of these summands
is a depth-$3$ $\Sigma\Pi\Sigma$ formula, in which every product gate is a product of linear forms; (iii) and \emph{each of these
linear forms consists of only a single variable}.

\textit{Note that any such formula $ \Phi $ is also an ordered formula}, since the products are of linear forms, each of a single variable, one can order the products in a way that respects the underlying variable order $ \preceq $.
\end{proof}

\subsubsection{Corollaries: short proofs and separations}\label{sec:cor_for_fop}
For natural numbers $ m>n $, denote by $\neg \textrm{FPHP}^m_n$ the following unsatisfiable  collection of polynomials:
\begin{equation}\label{eq037}
\begin{array}{ll}
 {\rm{Pigeons:}} & \forall i \in [m],\,\,\,(1- x_{i,1})  \cdots (1-x_{i,n}) \\
 {\rm{Functional:}} & \forall i \in [m]\,\forall k<\ell \in [n],\,\,\,x_{i,k}\cd x_{i,\ell}\\
 {\rm{Holes:}} & \forall i < j \in [m]\,\forall k \in [n], \,\,\,\,x_{i,k}\cd x_{j,k} \\
 \end{array}
\end{equation}\bigskip

As a corollary of the polynomial simulation of \RZ0 by \fop, and the upper bounds on \RZ0 proofs
demonstrated in \cite{RT07}, we get the following result:

\begin{corollary}\label{cor:FOP-PHP-ref}
For any linear order on the variables, and for any $m>n$ there are polynomial-size (in $n$) \fop~refutations of the $m$ to $n$ pigeonhole principle {\rm FPHP}$^m_n$ (over large enough fields).
\end{corollary}

$ \neg \textrm{FPHP}^m_n$ is a direct translation of the CNF formula for the $ m $ to $ n $ functional pigeonhole principle. Thus, by known lower bounds, \fop~is strictly stronger than resolution and is separated from bounded depth Frege.
On the other hand, Razborov \cite{Razb98} and subsequently Impagliazzo et al.~\cite{IPS99} gave exponential lower bounds on
the size of PC-refutations of a different \emph{low degree}
version of the Functional Pigeonhole Principle.
In this low degree version the Pigeons polynomials in (\ref{eq037}) are replaced by $\, 1- (x_{i, 1}+\ldots+
x_{i,n})$, for all $i\in[m]$.
It is not hard to show (via reasoning inside \RZ0) that \fop~admits polynomial-size refutations also for this low-degree version of the functional pigeonhole principle. This shows that \fop~is strictly stronger than PC (under the size measures as defined for \fop~and PC).\QuadSpace

The Tseitin graph tautologies were proved to be hard tautologies for several propositional proof system. We refer the reader to \cite{RT07}, Definition 6.5, for the precise definition of the (generalized, mod $ p $) Tseitin tautologies. We have the following:

\begin{corollary}\label{cor:fop_ub_Tseitin}
Let $G$ be an $r$-regular graph with $n$ vertices, where $r$ is a constant,
and fix some modulus $p$.
Then, for any linear order on the variables there are polynomial-size (in $n$) \fop~refutations of the corresponding
Tseitin {\rm mod} $p$ formulas \Tse0 (over large enough fields).
\end{corollary}

This stems from the \RZ0 polynomial-size refutations of the Tseitin {\rm mod} $p$
formulas demonstrated in \cite{RT07}.
From the known exponential lower bounds on PCR (and PC and resolution) refutation size of
Tseitin mod $p\,$ tautologies (when the underlying graphs are appropriately expanding; cf. \cite{BGIP01,BSI99,ABSRW00}), and for the polynomial simulation of PCR by \fop, we conclude  that \fop~is strictly stronger than PCR.

\section{Useful lower bounds on product of ordered polynomials}\label{sec:lb_on_product_formulas}
In this section we show that the ordered formula size of certain polynomials can increase exponentially when multiplying the polynomials together. We use this to suggest an approach
for lower bounding the size of \fop~proofs in Section \ref{sec:lower_bound_approach}.
We use a method of partial derivatives matrix introduced by Nisan to obtain exponential-size lower bounds on noncommutative formulas in \cite{Nis91}.

\begin{proposition}\label{prop:lower_bound_ordered_formula}
Let $ \F $  be a field, $X:= \set{x_1,\ldots,x_{n}} $ be a set of variables and $ \preceq $ be some linear order on $ X $.
Then, for any natural numbers $ m\le n $ and $ d \le \lfloor n/m \rfloor $, there exist polynomials $ f_1,\ldots,f_d $ from $ \F[x_1,\ldots,x_n] $, such that every $ f_i $ can be computed by an ordered formula of size $ O(m)$ and every ordered formula computing
$ \prod_{i=1}^d f_i $ has size $ 2^{\Omega(d)}$.
\end{proposition}

\begin{proof}
First, note that it is sufficient to prove the proposition for $ m=2 $ and any $ d\le \lfloor n/2\rfloor $.
(Because, assume that the proposition holds for $ m=2 $ and any $ d\le \lfloor n/2\rfloor $.
And let $ m',d' $ be such that  $ m'\le n $ and $ d'\le \lfloor n/m' \rfloor $.
By assumption, for $ m=2 $ and $ d'\le \lfloor n/m' \rfloor \le \lfloor n/2\rfloor$,
there are $ f_1,\ldots,f_{d'} $ from $ \F[x_1,\ldots,x_n] $ that can be computed by ordered formulas of size constant
[that is, $ O(2) $, and hence of size $ O(m')$], and such that every ordered formula computing $ \prod_{i=1}^{d'} f_i $ has size $ 2^{\Omega(d')} $.)

Thus, let $ m=2 $ and $ d \le\lfloor n/2 \rfloor $.
Assume without loss of generality that the linear order $ \preceq $ is such that $ x_1\preceq x_2\preceq\ldots\preceq x_n $. Abbreviate the variables $ x_1,\ldots,x_d $\, as\, $ y_1,\ldots,y_d $, respectively, and abbreviate the variables $ x_{d+1},\ldots,x_{2d} $\, as\, $ z_1,\ldots,z_d $, respectively (that is, the $ y_i $'s and $ z_i $'s are just different notations for their corresponding $ x_i $ variables, introduced to simplify the writing). We thus have $ y_1\preceq\ldots\preceq y_d \preceq z_1\preceq\ldots\preceq z_d $.

For every $ i =1,\ldots,d $, define the following polynomial:
\[
    f_i := (y_i+z_i)\,.
\]
Define
\[
    \mathsf{HARD}_d := \prod_{i=1}^d f_i = \prod_{i=1}^d (y_i+z_i)\,.
\]
We show that every ordered formula of $\mathsf{HARD}_d$ (under $ \preceq $)
is of size at least $  2^{\Omega(d)} $.
Note that $\mathsf{HARD}_d$ is a homogenous and multilinear polynomial of degree $ d $.

Recall that $\llbracket \mathsf{HARD}_d \rrbracket $ is the \emph{noncommutative} polynomial obtained from $\mathsf{HARD}_d$ by ordering the products in every monomial in accordance to the linear order $ \preceq $.
By definition of ordered formulas, it suffices to lower bound the size of noncommutative formulas
computing $\llbracket \mathsf{HARD}_d \rrbracket $.
For this purpose we use a rank argument introduced in \cite{Nis91}.
Nisan defined the matrix $ M_k(f) $ associated with a noncommutative polynomial $ f $ as follows:
\begin{definition}[\cite{Nis91}]
Let $ f \in \F\langle x_1,\ldots,x_n\rangle $ be a noncommutative homogenous polynomial of degree $ d $.
For every $ 0\le k\le d $, we define $ M_k(f) $ to be a matrix of dimension $ n^{k}\times n^{d-k} $ as follows:
(i) there is a row corresponding to every degree $ k $ noncommutative monomial over the variables $\{x_1,\ldots,x_n\} $, and a column corresponding to every degree $ d-k $ noncommutative monomial over the variables $\{ x_1,\ldots,x_n\} $;
(ii) for every degree $ k $ monomial $ \mathscr M $ and every degree $ d-k $ monomial $ \mathscr N $,
the entry in $ M_k(f) $ on the row corresponding to $ \mathscr M $ and column corresponding to $ \mathscr N $ is the coefficient of the degree $ d $ monomial $ \mathscr M\cd \mathscr N $ in $ f $.
\end{definition}

\begin{theorem}[\cite{Nis91} Theorem 1]\label{thm:Nisan_lower_bound}
Let $ f $ be a degree $ r$ homogenous noncommutative polynomial. Then, every noncommutative formula computing $ f $ has size at least $\sum_{k=0}^r{\rm rank}\left(M_k(f)\right)\,$.
\end{theorem}

In view of Theorem \ref{thm:Nisan_lower_bound}, it suffices to prove the following claim:
\begin{claim} For any $ 0\le k \le d $ we have
$  \,{\rm rank}(M_{k}(\llbracket \mathsf{HARD}_d\rrbracket)) \ge {d\choose k}\,.
 $
\end{claim}
\begin{proofclaim}
Consider the matrix $ M_k(\llbracket \mathsf{HARD}_d\rrbracket) $.
Let  $ \mathbf A_k $ be the matrix obtained from $ M_k(\llbracket \mathsf{HARD}_d\rrbracket) $ by removing  all rows and columns excluding the following rows and columns:
\begin{enumerate}
\item the rows corresponding to degree $ k $ multilinear monomials containing only $ y_i $ variables, such that the order of products in the monomial respects $ \preceq $\,;
\item the columns corresponding to degree $ d-k $ multilinear monomials containing only $ z_i $ variables, such that the order of products in the monomial respects $ \preceq $.
\end{enumerate}

Consider a degree $ k $ monomial $ \mathscr M = y_{i_1}\cdots y_{i_k} $, where $ i_1 <\ldots <i_k $. Let $ J = [d]\sm \set{i_1,\ldots,i_k} $. We can denote the elements of $ J $ as $ \set{j_1,\ldots,j_{d-k}} $, where $ j_1<\ldots<j_{d-k} $. Observe that the monomial  $\mathscr M $  has on its corresponding row in $ \mathbf A_k $ only zeros, except for a single $ 1 $ in the position (that is, column) corresponding to the degree $ d-k $ monomial $ \mathscr N = z_{j_1}\cdots z_{j_{d-k}} $. (Indeed, note that the coefficient of the degree $ d $ monomial $ \mathscr M\cd\mathscr N $ in $\llbracket \mathsf{HARD}_d\rrbracket$ is $ 1 $.)

Note that $ \mathbf A_k $ contains $ d \choose k $ rows corresponding to all possible degree $ k $ multilinear monomials $ \mathscr M $ in the $ \bar y $ variables whose product order respect $ \preceq $.
Similarly, $ \mathbf A_k $ contains $ d \choose k $ columns corresponding to all possible degree $ d-k $ multilinear monomials $ \mathscr N $ in the $ \bar z $ variables whose product order respect $ \preceq $.
By the previous paragraph: (i) each of the rows in $ \mathbf A_k $ has only one nonzero entry;
and (ii) for every row, the nonzero entry is in a \emph{different} column from those of other rows.
We then conclude that $ \mathbf A_k $ is a permutation matrix.
Therefore:
\[
{\rm rank}(\mathbf A_k) = {d \choose k}\,.
\]
The claim follows since clearly 
$ {\rm rank}(\mathbf  A_k) \le {\rm rank}(M_{k}\left(\llbracket\mathsf{HARD}_d\rrbracket)\right)\,$.
\end{proofclaim}

By the claim and by Theorem \ref{thm:Nisan_lower_bound}, we conclude that the ordered formula size of $\mathsf{HARD}_d$ is at least
\[
    \sum_{k=0}^d{\rm rank}\left(\mathbf A_k\right) =
        \sum_{k=0}^d{d \choose k} = 2^d\,.
\]
\end{proof}
\vspace{10pt}

\subsection{A lower bound approach}\label{sec:lower_bound_approach}
Here we discuss a simple possible approach intended to establish lower bounds on \fop\ proofs, roughly, by reducing \fop\ lower bounds to PC degree lower bounds and using the bound in
Section \ref{sec:lb_on_product_formulas} (Proposition \ref{prop:lower_bound_ordered_formula}).

Let $Q_1(\bar x),\ldots,Q_m(\bar x)$ be a collection of constant degree (independent of $ n $) polynomials from $ \F[x_1,\ldots,x_n] $ with no common solutions in $ \F $,
such that $ m $ is polynomial in $ n $. Let $ f_1(\bar y),\ldots,f_n(\bar y) $ be $ m $ homogenous polynomials of the same degree from $ \F[y_1,\ldots,y_\ell] $, such that the ordered formula size of each $ f_i(\bar y) $ (for some fixed linear order on the variables) is polynomial in $ n $ and such that the $ f_i(\bar y)$'s  do not have common variables (that is, each $ f_i(\bar y) $ is over disjoint set of variables from $ \bar y $).
Suppose that for any distinct $ i_1,\ldots,i_d \in [n]$  the ordered formula size of $ \prod_{j}^d f_{i_j}(\bar y) $ is $ 2^{\Omega(d)} $.

\begin{note}
By the proof of Proposition \ref{prop:lower_bound_ordered_formula}, the conditions above are easy to achieve. Indeed, the $ f_i(y_i,z_i) $'s defined in the proof of Proposition \ref{prop:lower_bound_ordered_formula} have these properties: homogeneity, same degrees for all $ f_i $'s and disjointness of variables, and an exponential increase in ordered formula size for products of the $ f_i $'s.
\end{note}

Consider the polynomials $ Q_1(\bar x),\ldots,Q_m(\bar x) $ after applying the substitution:
\begin{equation}\label{eq:substitution_itself}
x_i \mapsto f_i(\bar y)\,.
\end{equation}

In other words, consider
\begin{equation}\label{eq:substitution_hard_instance}
Q_1(f_1(\bar y),\ldots,f_n(\bar y)),\ldots,Q_m(f_1(\bar y),\ldots,f_n(\bar y))\,.
\end{equation}

Note that (\ref{eq:substitution_hard_instance}) is also unsatisfiable over $ \F $.
We suggest to lower bound the \fop\ refutation size of (\ref{eq:substitution_hard_instance}),
based on the following simple idea:
it is known that some families of unsatisfiable collections of polynomials require linear $\Omega(n) $ \emph{degree} PC refutations (where $ n $ is
the number of variables). In other words, every refutation of these polynomials must contain some polynomial of linear degree. By definition, also every \fop~refutation of
these polynomials must contain some polynomial of linear degree.

Thus, assume that the initial polynomials $Q=\set{ Q_1(\bar x),\ldots,Q_m(\bar x)}$ in the $ x_1,\ldots,x_n $ variables, require linear degree refutations---in fact, an $ \omega(\log n)$ degree lower bound would suffice.
Thus, every PC refutation contains some polynomial $ h $ of degree $\omega(\log n) $.
Then, we might expect that every PC refutation of its substitution instance (\ref{eq:substitution_hard_instance}) contains a polynomial $ g\in \F[\bar y] $ which is a substitution instance (under the substitution (\ref{eq:substitution_itself})) of an $ \omega(\log n) $-degree polynomial in the $\bar x $ variables. This, in turn, leads (under some conditions; see below for an example of such conditions) to a lower bound on \fop\ refutations.
An example of sufficient conditions for super-polynomial \fop~lower bounds, is as follows: every PC refutation of (\ref{eq:substitution_hard_instance}) contains a polynomial $ g $ so that one of $ g $'s homogenous components is a substitution instance (under the substitution (\ref{eq:substitution_itself})) of a degree $ \omega(\log n) $ multilinear polynomial from $ \F[x_1,\ldots,x_n] $.
We formalize this argument:
\vspace{10pt}

\begin{example} \textbf{conditional \fop\ size lower bounds.}
\label{prop:conditional_lb_fop}
(Assume the above notations and conditions.)
\ind\textbf{If:} every PC refutation of (\ref{eq:substitution_hard_instance}) that has polynomial in $ n $ number of proof-lines contains a polynomial $ g \in \F[y_1,\ldots,y_\ell]  $ such that for some $ t \le \deg(g)$, the $ t $-th homogenous component $ g^{(t)} $ of $ g $ (that is, the sum of all monomials of total degree $ t $ in $ g $) is a \emph{substitution instance} (under the substitution (\ref{eq:substitution_itself})) of a degree $ \omega(\log n) $ multilinear polynomial from $ \F[x_1,\ldots,x_n] $;

\ind\textbf{Then:} every \fop~refutation of (\ref{eq:substitution_hard_instance}) is of super-polynomial size (in $ n $).
\end{example}

\begin{proof_of_example}
It suffices to show that any ordered formula of $ g $ is of super-polynomial size in $ n $. Note that breaking an algebraic formula into its corresponding homogenous components---according to the standard known procedure (cf. \cite{Raz08}, proof of Proposition 2.3)---is also applicable to ordered formulas: in other words, if $ g $ has a polynomial-size ordered formula then each of $ g $'s homogenous components has a polynomial-size ordered formula as well.\footnote{Assume we have an ordered formula $ \Phi $ and we want to construct the ordered formula $ \Phi^{(k)}$ that computes the $ k $-th degree homogenous polynomial of (the polynomial computed by) $ \Phi $. We work by induction on the structure of the formula $\Phi $: a \emph{plus gate} in the original formula $\Phi $ turns into a plus gate $ u $ with two children, such that if each of the two subformulas rooted at the two children are ordered formulas then the subformula rooted at $ u $ is also an ordered formula. A \emph{product gate} turns into the sum of products of pairs of ordered subformulas, such that if the original product gate respects the linear order then also each of the products in the sum respects the linear order.
(For more details on the construction of [non-ordered] homogenous formulas from a given algebraic formula we refer the reader to \cite{Raz08}.)}
Thus, it suffices to show that every ordered formula of $ g^{(t)} $ is of size super-polynomial in $ n $ (because then $ g $ itself has
super-polynomial size).

By assumption, $ g^{(t)} $ is a substitution instance of some degree $ \omega(\log n) $ multilinear polynomial $ h \in \F[x_1,\ldots,x_n] $.
Since $ g^{(t)} $ is homogenous and all the $ f_i(\bar y) $'s have the same degree and are homogenous, $ h $ must be homogenous too.
Since $ h $ is multilinear we can write $ h = \sum_{j\in J} b_j \mathcal M_j $, where the $ \mathcal M_j $'s are multilinear monomials in the $ \bar x $ variables and $ b_j $ are coefficients from $ \F $.
Now, consider some single monomial $ \mathcal M $ from  $ \sum_{j\in J} b_j \mathcal M_j $.
By multilinearity and homogeneity of $ h $ every other monomial $ \mathcal M'\neq \mathcal M $ in $ h $ must contain an $ x_i $ variable that does not appear in $ \mathcal M $. We can assign $ 0 $ to such $ x_i $.
Doing this for every monomial $ \mathcal M'\neq \mathcal M $, we get that $ h $ (under this partial assignment to the $ \bar x $ variables) is equal to $ b \mathcal M $, for some coefficient $ b \in \F $.
In a similar manner, by disjointness of the variables in the $ f_i(\bar y) $'s, there exists a partial assignment $\rho:\bar y \to \set{0} $, such that $ g^{(t)}\rst\rho $ is just a substitution instance (under the substitution (\ref{eq:substitution_itself})) of a single
degree $ \omega(\log n) $ multilinear \emph{monomial} in the $\bar x $ variables.
This means that $ g^{(t)}\rst\rho $ is the product of $\omega(\log n) $ distinct $ f_i(\bar y) $'s (multiplied by $ b $).
Therefore, by assumption on the $ f_i(\bar y) $'s every ordered formula of $ g^{(t)} $ is of size exponential in $ 2^{\omega(\log n)} $, which is super-polynomial in $ n $.
\end{proof_of_example}

\section*{Acknowledgments}
I wish to thank Emil je\v{r}abek, Sebastian M\"{u}ller,  Pavel Pudl{\'a}k and Neil Thapen for helpful discussions on issues related to this paper. I also wish to thank Ran Raz for suggesting this research direction, and Jan Kraj\'{i}\v{c}ek for inviting me to give a talk at TAMC 2010 on
this subject.

\bibliographystyle{alpha} 

\newcommand{\etalchar}[1]{$^{#1}$}
\def\cprime{$'$} \def\cprime{$'$} \def\cprime{$'$} \def\cprime{$'$}

\HalfSpace
\end{document}